\documentclass[reqno,english,empty, 12pt]{amsart}
\usepackage{array}
\usepackage[colorlinks=true,
 linkcolor=blue, citecolor=blue,
]{hyperref}
\usepackage[utf8]{inputenc}

\usepackage[T1]{fontenc}

\usepackage[a4paper, margin=2.7cm]{geometry}
\usepackage{enumerate}
\usepackage{amsmath, amsthm, amsfonts, amssymb}
\usepackage{mathrsfs}           

\usepackage{bm}

\newcommand{\abs}[1]{\left\vert#1\right\vert}

\newcommand{\br}[1]{\left\lbrace #1\right\rbrace}
\newcommand{\pa}[1]{\left( #1 \right)}
\newcommand{\re}{\exp}

\def\N{\mathbb{N}}
\def\R{\mathbb{R}}

\def\Q{\mathcal{Q}}

\def \E {\mathcal{E}}
\def\A{\mathcal{A}}
\def\1{\mathbbm{1}}

\def\d{\,\mathrm{d}}

\def \ddt{\frac{\mathrm{d}}{\mathrm{d}t}}

\numberwithin{equation}{section}

\def\:{\colon}

\newtheorem{thm}{Theorem}[section]
\newtheorem{cor}[thm]{Corollary}
\newtheorem{lem}[thm]{Lemma}
\newtheorem{prp}[thm]{Proposition}

\theoremstyle{definition}
\newtheorem{dfn}[thm]{Definition}
\theoremstyle{remark}
\newtheorem{rem}[thm]{Remark}
\theoremstyle{example}

\def \leq {\leqslant}
\def \geq {\geqslant}

\title{Uniform moment propagation for the Becker-Döring equation}

\author{Jos\'e A. Ca\~nizo, Amit Einav and Bertrand Lods}

\address{Departamento de Matem\'{a}tica
  Aplicada, Universidad de Granada, Av. Fuentenueva S/N, 18071
  Granada, Spain}
\email{canizo@ugr.es}

\address{Institut f\"ur Analysis und Scientific Computing,
  Technische Universit\"at Wien, Wiedner Hauptstrasse 8-10
  A-1040 Wien, \"Osterreich}
\email{aeinav@asc.tuwien.ac.at}
\address{Departement of Economics and Statistics  \& Collegio
  Carlo Alberto, Universit\`{a} degli Studi di Torino,  Corso Unione
  Sovietica, 218/bis, 10134 Torino, Italy}
\email{bertrand.lods@unito.it}

\begin{document}

\begin{abstract}
  We show uniform-in-time propagation of algebraic and stretched
  exponential moments for the Becker-Döring equations. Our proof is
  based upon a suitable use of the maximum principle together with
  known rates of convergence to equilibrium.
\end{abstract}
 
\maketitle
\tableofcontents

\section{Introduction}
\label{sec:intro}

In this note we consider  the Becker-Döring equations 
\begin{subequations}
  \label{eq:BD}
  \begin{align}
     \label{eq:BD1}
    \frac{\d}{\d t} c_i(t) &= W_{i-1}(t) - W_{i}(t), \qquad i\in \N\setminus \br{1},
    \\
    \label{eq:BD2}
    \frac{\d}{\d t} c_1(t) &= - W_1(t) - \sum_{k=1}^\infty W_{k}(t),
 \end{align}
\end{subequations}
where
\begin{equation}
  \label{eq:def-Wi}
  W_{i}(t) := a_{i}\, c_1(t) c_i(t) - b_{i+1}\, c_{i+1}(t)
  \qquad i \in \N.
\end{equation}

The unknowns here are the functions
$\bm{c}(t) = \left(c_{i}(t)\right)_{i\geq 1}$ which depend on time
$t \geq 0$ and where, for each $i \in \N$, $c_{i}(t)$ represents the
density of clusters of size $i$ at time $t \geq 0$ (that is, clusters
composed of exactly $i$ individual particles). The non-negative numbers
$a_{i}, b_{i}$ denote respectively the coagulation and fragmentation
coefficients. These equations are a model for the dynamics of cluster
growth in which clusters can only gain or shed one particle; that is,
the only reactions taking place are
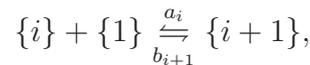
\begin{equation*}
  \{i\} + \{1\} \underset{b_{i+1}}{\overset{a_i}{\leftrightharpoons}} \{i+1\},
\end{equation*}
where $\{i\}$ represents the concentration of clusters of size
$i$. The quantity $W_i$ then represents the net rate of this reaction,
obtained by standard mass-action dynamics. It is a well accepted model
for the kinetics of first order phase transitions, applicable to a
wide variety of phenomena such as crystallisation, vapor condensation,
aggregation of lipids or phase separation in alloys. The model is
traced back to \cite{BD35}, and the basis of its mathematical theory
was set in \cite{Ball1988Asymptotic,BCP86}. There have been a number
of works on the long-time behaviour of solutions, which is especially
interesting since it exhibits phase-change phenomena, metastability,
and fast relaxation to equilibrium depending on the regime one is
considering. We mention here the works by \cite{Canizo2013Exponential,Canizo2015Trend,JN03,MP2016,MP2017,N08,Pen97,
  Penrose1989Metastable,Sch16,Velazquez1998BeckerDoring},
leaving out many relevant ones. We direct the reader to the references
in the aforementioned works for a more complete picture, and to the
survey paper \cite{slemrod}.

Despite the amount of works devoted to the model, it seems to us that
the question of propagation of moments has not been fully answered,
and it is our purpose to fill that gap in this paper. The basic
question we address is the following: if
$\sum_{i=1}^\infty i^k c_i(0) < +\infty$ is finite for some $k > 1$,
is it true that $\sum_{i=1}^\infty i^k c_i(t) \leq C$ for some $C > 0$
and all $t \geq 0$? We show an affirmative answer for subcritical solutions, which is the natural case in which one expects it to hold.

Before describing our results with more detail we need to set some
notation and give some background on the asymptotic behaviour of
equation \eqref{eq:BD}.

\subsection{A quick summary on asymptotic behaviour}
\label{sec:prelim}

Equation \eqref{eq:BD} can be written, in weak form, as
\begin{equation}
  \label{eq:BD-weak}
  \ddt \sum_{i=1}^\infty c_i(t) \phi_i
  = \sum_{i=1}^\infty W_i(t) (\phi_{i+1} - \phi_i - \phi_1),
\end{equation}
for all slowly growing sequences $(\phi_i)_{i \geq 1}$. In particular,
taking $\phi_{i}=i$, one sees that the \emph{density} of the
solution, defined by
\begin{equation}
  \label{eq:def-density}
  \varrho :=
  \sum_{i=1}^\infty i c_i(0)
  = \sum_{i=1}^\infty i c_i(t),
\end{equation}
is formally conserved under time evolution. Defining the
\emph{detailed balance coefficients} $Q_i$ recursively by
\begin{equation}
  \label{eq:Qi}
  Q_1 = 1,
  \quad
  Q_{i+1} = \frac{a_i}{b_{i+1}} Q_i
  \quad i\in\N
\end{equation}
one can see that any sequence of the form $(Q_{i}z^{i})_{i\geq 1}$ is
formally an equilibrium of \eqref{eq:BD}. However, such a sequence may
not have a finite density. The largest $z_s \geq 0$ (possibly
$z_s=+\infty$) for which
\begin{equation}
  \label{eq:zs-def}
  \sum_{i=1}^{\infty} iQ_i z^i < +\infty
  \quad \text{ for all $0\leq z <z_s$}
\end{equation}
is called the \emph{critical monomer density}, or sometimes the
monomer saturation density (alternatively, $z_s$ is the radius of
convergence of the power series with coefficients $i Q_i$). The
\emph{critical density} (or, again, saturation density) is then
defined by
\begin{equation*}
  \label{eq:critical-mass}
  \varrho_s := \sum_{i=1}^\infty i Q_i z_s^i \in [0,+\infty].
\end{equation*}
This critical density plays a fundamental role in the long-time
behaviour of solutions to \eqref{eq:BD}: it was proved in \cite{BCP86}
and \cite{Ball1988Asymptotic} that any solution with density
$\varrho > \varrho_{s}$ will converge (in a weak sense) to the only
equilibrium with density $\varrho_{s}$, with the excess mass
$\varrho - \varrho_s$ becoming concentrated in larger and larger clusters as
time passes. In contrast, any solution with initial density
$\varrho \leq \varrho_{s}$ will converge (strongly) as $t \to \infty$
to an equilibrium solution with its same density $\varrho$. We focus
here on the so-called \emph{subcritical solutions} for which
$\varrho < \varrho_{s}$, which converge to the equilibrium
$\bm{\Q}:=\left(\Q_{i}\right)_{i\geq1}$ given by
$$\Q_{i}=Q_{i}\overline{z}^{i}, \qquad i\geq 1,$$ where
$\overline{z} \in [0,z_{s})$ is the unique number such that
$\varrho=\sum_{i=1}^\infty Q_i \overline{z}^i.$ The rate of
convergence to this equilibrium in exponentially weighted
$\ell_{1}(\N)$ norms was studied in \cite{JN03} and subsequently
improved in \cite{Canizo2013Exponential}. Convergence for solutions
with finite algebraic moments (which applies to a wider range of
initial conditions) has been studied in \cite{Canizo2015Trend,
  MP2016, MP2017}.

The approach in \cite{JN03} is based on the entropy-entropy production
method and has been recently revisited by the authors of the present
paper in \cite{Canizo2015Trend}. It consists in estimating in a
careful way the evolution of the \emph{relative free energy}
\begin{equation}
  \label{eq:relative-free-energy}
  H(\bm{c}(t) | \bm{\Q}) :=  \sum_{i=1}^\infty \left(c_{i}(t)	\log \frac{c_i(t)}{\Q_i} - c_{i}(t)+\Q_{i} \right), \qquad t \geq0.
\end{equation}
We observe that $H(\bm{c}(t) | \bm{\Q})$ is finite whenever the
solution $\bm{c}(t)=(c_i(t))_{i \geq 1}$ is nonnegative and has finite
density (see for example Lemma 7.1 and 7.2 in \cite{Canizo2007Convergence}). We refer to \cite{Canizo2015Trend} for
more details on the entropy-entropy production method in the context
of the Becker-D\"oring equations.

\subsection{Main results}

A fundamental tool in the application of the entropy method is a
\emph{uniform control} of suitable moments of the solution $\bm{c}(t)$
to \eqref{eq:BD}, i.e. the control of suitable weighted-$\ell_{1}(\N)$
estimates. For instance, the analysis of \cite{JN03} deals with
subcritical solutions with finite exponential moments and is based on
the property that
\begin{equation}\label{eq:expon}
  \sum_{i=1}^\infty \re\pa{\eta i} c_i(0) < +\infty \implies \sup_{t
    \geq 0}\sum_{i=1}^{\infty}\re\pa{\eta'\,i}c_{i}(t) < \infty
\end{equation}
for $\eta > 0$ and some $0 < \eta'< \eta$. This was proved in
\cite{JN03}, and is to our knowledge the only available result on
uniform propagation of moments.

We would like to have a similar information for \emph{algebraic
  moments} $\sum_{i \geq 1} i^k c_i(t)$ or \emph{stretched exponential
  moments} of the form $\sum_{i \geq 1} \exp\pa{\alpha i^\mu} c_i$, for
some $\alpha > 0$ and $0 < \mu < 1$. Propagation of these moments on a
finite time interval is known to hold from the results of \cite{BCP86}
(see Lemma \ref{lem:short-time} hereafter), but the estimate on the
time interval $[0,T]$ deteriorates as $T$ increases. We intend to fill
this blank with the uniform in time propagation results in the next
two theorems. For our results we assume that
\begin{equation}
  \label{eq:ai1}
  \begin{aligned}
    \text{either}
    \qquad
    &0 < a_i \leq \overline{a}\,i^\gamma
    \qquad
    &&\text{for all $i \geq 1$ and some $\overline{a} >
      0$,\quad $0 \leq \gamma < 1$}
    \\
    \text{or}
    \qquad
    &C_1 i \leq a_i \leq C_2 i
    \qquad
    &&\text{for all $i \geq 1$ and some $0 < C_1 \leq C_2$.}
  \end{aligned}
\end{equation}
If the second option holds we call $\gamma = 1$ for consistency. We
also assume that
\begin{equation}
  \label{eq:bi}
  0 < b_{i} \leq \overline{b}\,a_{i},
\end{equation}
for all $i \geq 1$ and some $\overline{b} > 0$, and that
\begin{equation}
  \label{eq:Qi-limit}
  \lim_{i \to +\infty} \frac{Q_{i+1}}{Q_{i}} = \frac{1}{z_s}
  \qquad \text{for some $0 < z_s < +\infty$.}
\end{equation}
(Note that $z_s$ is indeed the critical monomer density, in agreement
with \eqref{eq:zs-def}.)  Additionally we may assume that the critical
equilibrium is non-increasing:
\begin{equation}
  \label{eq:critical-eq}
  \text{The sequence $\{Q_i z_s^i\}_{i}$ is non-increasing,}
\end{equation}
though this is not a fundamental requirement and small changes can be
made to adapt the proofs if the sequence
$\{Q_i z_s^i\}_{i \geq i_{0}}$ is non-increasing only for some fixed
$i_0 \in \N$. \\
Assumptions \eqref{eq:ai1}-\eqref{eq:critical-eq} are
  natural and satisfied in most physically relevant situations. It is
  worth mentioning that some of the assumptions can follow from
  others, with additional conditions. For instance, if one assumes
  that $\frac{b_{i+1}}{b_i}$ is bounded from below then the definition
  of $Q_i$ and assumption \eqref{eq:Qi-limit} imply that \eqref{eq:bi}
  is satisfied. Common coefficients that appear in the theory of
  density conserving phase transitions \cite{BCP86,N08} are
  $$a_i=i^\gamma,\quad\quad b_i=a_i\pa{z_s+\frac{q}{i^{1-\mu}}},$$
  for some $0<\gamma \leq 1$, $q>0$, and $0<\mu<1$. A different
  modelling assumption yields the coefficients
  $$a_i=i^\gamma,\quad\quad b_i=z_s\pa{i-1}^\gamma e^{\sigma i^\mu-\sigma(i-1)^\mu},$$
  for some $0<\gamma\leq 1, \sigma>0$ and $0<\mu<1$. It is easy to
  verify that these type of coefficients satisfy all our assumptions.
Regarding the initial datum
$\bm{c}^0 = (c_i^0)_{i \geq 1}$, we assume it is non-negative and has
some finite moments:
\begin{equation}
  \label{eq:c0-moments}
  \sum_{i=1}^\infty i^r c_i^0 < +\infty
  \qquad \text{for $r = \max\{2-\gamma, 1+\gamma\}$.}
\end{equation}
With this at hand, we can now state our first result:
\begin{thm}[\textit{\textbf{Uniform propagation of moments}}]
  \label{thm:moment-propagation}
  Assume \eqref{eq:ai1}--\eqref{eq:critical-eq}, and let
  $\bm{c}(t) = (c_i(t))_{i \geq 1}$ be a solution to the Becker-Döring
  equations \eqref{eq:BD} with non-negative, subcritical initial datum
  $\bm{c}(0)$ and density $\varrho < \varrho_s$. Let
  $k \geq \max\{2-\gamma, 1 + \gamma\}$ be such that
  \begin{equation*}
    M_k(0) := \sum_{i=1}^\infty i^k c_{i}(0)  < \infty.
  \end{equation*}
  There exists a constant $C > 0$ depending only on $k$, $M_k(0)$, the
  density $\varrho$ and the coefficients $(a_i)_{i \geq 1}$,
  $(b_i)_{i \geq 1}$ such that
  \begin{equation*}
    M_k(t) := \sum_{i=1}^\infty i^k c_i(t)
    \leq
    C
    \qquad \text{for all $t \geq 0$.}
  \end{equation*}
\end{thm}

The constant $C$ can be estimated explicitly from the proof. Our
second result deals with the uniform propagation of stretched
exponential moments in a similar way:

\begin{thm}[\textit{\textbf{Uniform propagation of stretched
      exponential moments}}]
  \label{thm:moment-propagation2}
  Assume \eqref{eq:ai1}--\eqref{eq:critical-eq} hold, with the first
  option in \eqref{eq:ai1} being true (for some $0 \leq \gamma <
  1$).
  Let $\bm{c}(t) = (c_i(t))_{i \geq 1}$ be a solution to the
  Becker-Döring equations \eqref{eq:BD} with non-negative, subcritical
  initial datum $\bm{c}(0)$ and density $\varrho$. Let
  $0 < \mu \leq 1-\gamma$ and $\alpha > 0$ be such that
  \begin{equation*}
    \E_\mu(0) := \sum_{i=1}^\infty \re\pa{\alpha i^{\mu}} c_{i}(0)  < \infty.
  \end{equation*}
  There exists a constant $C > 0$ depending only on $\mu$, $\alpha$,
  $\E_\mu(0)$, the density $\varrho$ and the coefficients
  $(a_i)_{i \geq 1}$, $(b_i)_{i \geq 1}$ such that
  \begin{equation*}
    \E_\mu(t) := \sum_{i=1}^\infty \re\pa{\alpha i^{\mu}} c_i(t)
    \leq
    C
    \qquad \text{for all $t \geq 0$.}
  \end{equation*}
\end{thm}

Notice that these two results prove \emph{uniform propagation} of the
considered moments whereas the result of \cite{JN03} recalled in
\eqref{eq:expon} is of a slightly different nature because of the
deterioration of the constant $\eta$ which measures the strength of
the exponential. We do not know whether uniform propagation is true
for exponential moments (that is, we do not know whether one can take
$\eta' = \eta$ in \eqref{eq:expon}); our method does not immediately
apply in this case since the short-time propagation in Lemma
\ref{lem:short-time} does not apply to exponential moments.

We mention here that, besides its own interest, Theorem
\ref{thm:moment-propagation} plays a crucial role in the determination
of the convergence rate to equilibrium for solutions to \eqref{eq:BD}
recently established in \cite{Canizo2015Trend}.

\subsection{Method of proof}

A natural attempt to prove the above results would be to directly
compute the evolution of $M_{k}(t)$ or $\E_{\mu}(t)$. Namely, picking
$\phi_{i}=i^{k}$ in the weak form \eqref{eq:BD-weak}, we get the
evolution of $M_{k}(t)$
\begin{equation*}
  \dfrac{\d}{\d t}M_{k}(t)
  =
  \sum_{i=1}^{\infty}\left(a_{i}c_{1}(t)c_{i}(t)-b_{i+1}c_{i+1}(t)\right)
  \left((i+1)^{k}-i^{k}-1\right),
\end{equation*}
and one may try to obtain a suitable differential inequality for $M_k$
in the spirit of similar results for kinetic equations (see
\cite{ACGM} for an example on the Boltzmann equation). This method is
rather efficient to obtain local in time bounds on $M_{k}(t)$ (or
$\E_{\mu}(t)$) but seems difficult to apply to get uniform bounds on
$[0,\infty)$. The difficulty stems from the fact that the ``loss
term'' $b_{i+1}c_{i+1}(t)$ appearing in the evolution does not always
compensate the ``gain term'' $a_{i}c_{1}(t)c_{i}(t)$. A deeper reason
for this is that boundedness of moments must depend on the mass of the
solution (since moments are never uniformly bounded for supercritical
solutions), so any estimate that gives uniform bounds must somehow
involve the mass of the solution. In practice, it is the value of
$c_1(t)$ that appears when one tries to bound the time evolution of
moments, and any uniform estimate seems to require some a priori
knowledge on the behaviour of $c_1(t)$. This is in contrast with the
situation for the Boltzmann equation (with hard potential
interactions) where the optimal Povzner's inequality allows us to
control the contribution of gain part of the collision operator by
that of its loss part (see for example \cite{Bobylev1996Moment}). It
should be remarked that the behaviour of moments for the Boltzmann
equation does not depend on the mass of the initial datum, but only on
which moments are initially finite, which is a fundamental difference
with the present case. Another important difference is the fact that
there is no creation of moments (of any kind) for the Becker-D\"oring
equations (see \cite{BCP86}).

We adopt here a different approach relying on a maximum principle. A
crucial role in our study will be played by the \emph{tail density}
$\mathcal{G}(t)=\left(G_{j}(t)\right)_{j\geq 1}$ given by
$$G_{j}(t)=\sum_{i=j}^{\infty}c_{i}(t), \qquad j \geq 1.$$
The main properties of $\mathcal{G}$ which are relevant for us are
established in Lemmas \ref{lem:propertiesofpartial} and
\ref{cor:diffinequalityforG}. Tail density was already introduced in
\cite{LPMS} in order to establish uniqueness of the solution to
\eqref{eq:BD} and a variant of it was used in
\cite{Canizo2005Asymptotic} to show strong convergence to equilibrium
for a generalised discrete coagulation--fragmentation model.

It is important to notice that moments of $\bm{c}(t)$ can be estimated
by suitable moments of $\mathcal{G}(t)$, so that Theorems
\ref{thm:moment-propagation}--\ref{thm:moment-propagation2} can be
stated in terms of moments of $\mathcal{G}(t)$ (the rough idea being
that the $k$-th moment of $\bm{c}$ is equivalent to the $(k-1)$-th
moment of $\mathcal{G}$; see Lemma \ref{lem:propertiesofpartial}). Of
course, the main interest is that the equation solved by
$\mathcal{G}(t)$ is somewhat simpler than \eqref{eq:BD}: one has
$$\ddt G_{j}(t)= a_{j-1}c_1(t)
\pa{G_{j-1}(t)-G_j(t)}+b_j \pa{G_{j+1}(t)-G_{j}(t)} \qquad j \geq 2.$$
The evolution equation for $\mathcal{G}(t)$ depends on $c_1(t)$, and
the entire nonlinear structure of the interaction between clusters is
driven by it (assuming $c_{1}(t)$ to be known in \eqref{eq:BD} would
yield a linear system of ODEs). Since the coefficient of $c_1(t)$ is
non-negative in the above equation, if one is able to control
$c_{1}(t)$ from above on some given interval, then one can bound the
above evolution of $\mathcal{G}(t)$ by a suitable infinite system of
differential inequalities, represented by an infinite matrix whose
off--diagonal entries are non-negative. This is the key ingredient
that yields a maximum principle for the evolution of $\mathcal{G}(t)$
(see Lemmas \ref{cor:diffinequalityforG} and \ref{lem:max-ode}). The
proof then consists in establishing the existence of suitable
supersolutions to the Becker-D\"oring equations whose moments are
strongly related to the moments of $\mathcal{G}(t)$ in order to apply
the maximum principle. As already said, this will be possible once a
suitable bound on $c_{1}(t)$ has been established. To prove an a
priori bound for $c_1(t)$ we resort to general results of
\cite{Canizo2007Convergence} (when $\gamma < 1$) and
\cite{Canizo2015Trend} (when $a_i \sim i$) where the rate of
convergence to equilibrium for solutions to \eqref{eq:BD} has been
established under mild assumptions on the initial data. Notice that
the rate obtained in \cite{Canizo2007Convergence} is far from being
optimal but applies to a wide range of initial data, and ensures at
least the existence of some explicit time $T >0$ such that
$c_{1}(t) < z_{s}$ for $t \geq T$. This is enough to apply the method
we just described.

\subsection{Organization of the paper}

In the next section we introduce the main tools for the proof of both
Theorems \ref{thm:moment-propagation} and \ref{thm:moment-propagation2},
namely the introduction of the tail density $\mathcal{G}(t)$ and the
maximum principle. The proofs of Theorems
\ref{thm:moment-propagation} and \ref{thm:moment-propagation2} are then
given in Section 3 after recalling the result on convergence to
equilibrium in \cite{Canizo2007Convergence}.

\section{Tail density and the maximum principle}
\label{sec:maximum}

A key idea for showing our main theorems is to find a quantity, which
we will call \emph{the tail density}, that obeys a maximum principle
for the equation and whose moments are intimately connected to the
moments of $\bm{c}(t)$, the solution to the Becker-D\"oring equations.
\begin{dfn}\label{def:partialdensity}
  Let $\bm{c}=\br{c_i}_{i\in\N}$ be a non-negative, summable
  sequence. We define the \emph{tail density of $\bm{c}$} as the
  sequence $\mathcal{G}=\br{G_{j}}_{j\in \N}$ given by
  \begin{equation}\label{eq:defpartial}
    G_j = \sum_{i=j}^\infty c_i,
    \qquad j \in \N.
  \end{equation}
\end{dfn}
The tail density enjoys the following properties:
\begin{lem}\label{lem:propertiesofpartial}
Let $\bm{c}=\br{c_i}_{i\in\N}$ be a non-negative, summable sequence. Then
\begin{enumerate}[(i)]
\item the tail density $\mathcal{G}=\br{G_j}_{j\in\N}$ is a non-negative, non-increasing sequence.
\item For any $k \geq 0$
\begin{equation}\label{eq:momentequivalence}
\frac{M_{k+1}\pa{\bm{c}}}{k+1} \leq M_{k}\pa{\mathcal{G}} \leq M_{k+1}\pa{\bm{c}}.
\end{equation}
\item Given $\gamma \in [0,1)$, let $\alpha > 0$ and $\mu \in (0,1-\gamma).$ Introduce 
$$\E_{\mu}(\bm{c})=\sum_{i\geq1}\re\pa{\alpha\,i^{\mu}}c_{i}.$$ 
Then, there exist $\eta_1,\eta_2>0$, depending only on $\mu$ and $\alpha$, such that
\begin{equation}\label{eq:doublebound}  
  \eta_{1}\sum_{j=1}^{\infty} \psi_j G_{j}
  \leq
  \E_{\mu}(\bm{c})
  \leq
  \eta_{2}\sum_{j=1}^{\infty} \psi_j G_{j}
  \end{equation}
where $\psi_{j}:= j^{\mu-1} \re\pa{\alpha j^{\mu}},$ for all $j \geq 1.$
\end{enumerate}
\end{lem}

\begin{proof}

  Point $(i)$ is clear from the definition of the tail density. To
  show $(ii)$ we notice that
  $$\frac{i^{k+1}}{k+1}=\int_0 ^ i x^k dx \leq \sum_{j=1}^i j^k \leq i^{k+1}, \qquad 	\forall i \geq 1.$$
  Since
  $$\sum_{j=1}^\infty j^k G_j = \sum_{j=1}^\infty j^k \pa{\sum_{i=j}^\infty c_i}=\sum_{i=1}^\infty c_i \pa{\sum_{j=1}^i j^k},$$
  where we were allowed to change summation due to the non-negativity
  of the elements, the proof of $(ii)$ complete.

  To prove point $(iii)$ we write $c_i = G_i - G_{i+1}$ to obtain
  \begin{multline}
    \label{eq:Emu}
    \E_\mu(\bm{c})
    = \sum_{i \geq 1} \re\pa{\alpha i^\mu} c_i
    = \sum_{i \geq 1} \re\pa{\alpha i^\mu} (G_i - G_{i+1})
    \\
    = \re\pa{\alpha} G_1 + \sum_{i \geq 2} G_i \left(
      \re\pa{\alpha i^\mu} - \re\pa{\alpha (i-1)^\mu}
    \right).
  \end{multline}
  Since
  $$\re\pa{\alpha i^\mu} - \re\pa{\alpha (i-1)^\mu}
  \leq
  \alpha \mu (i-1)^{\mu-1} \re\pa{\alpha i^{\mu}},
  \qquad i\geq 2
  $$
  we have
  \begin{equation}\label{eq:stretchedproofI}
    \begin{split}
      \E_{\mu}(\bm{c})
      &\leq
      \re\pa{\alpha} G_1
      + \alpha \mu \sum_{i=2}^\infty (i-1)^{\mu-1}
      \re\pa{\alpha i^\mu}G_i
      \\
      &\leq
     \re\pa{\alpha} G_1
      +2^{1-\mu} \alpha \mu  \sum_{i=2}^\infty i^{\mu-1}
      \re\pa{\alpha i^\mu}G_i      \\
      &\leq \max\pa{ 1,\ 2^{1-\mu} \alpha \mu}
      \sum_{j=1}^\infty \psi_j G_j.
    \end{split}
  \end{equation}
  In addition, 
  $$\re\pa{\alpha i^\mu} - \re\pa{\alpha (i-1)^\mu}
  \geq
  \alpha \mu i^{\mu-1} \re\pa{\alpha (i-1)^{\mu}},
  \qquad i\geq 2.
  $$
  Thus, using the fact that
  $\re\pa{ \alpha j^\mu - \alpha (j-1)^\mu}
    \underset{j\rightarrow \infty}{\longrightarrow}1$
    when $0<\mu<1$, from \eqref{eq:Emu} we conclude that
  \begin{equation}
    \label{eq:stretchedproofII}
    \begin{gathered}
      \E_{\mu}(\bm{c})\geq \re\pa{\alpha} G_1
      +
      C
      \sum_{i=2}^\infty
      i^{\mu-1} \re\pa{\alpha i^\mu} G_i
      \geq
      \min\{1, C\}
      \sum_{j=1}^\infty \psi_j G_j,
    \end{gathered}
  \end{equation}
  where
  $C = \alpha \mu\, \inf_{j\geq 2} \re\pa{\alpha
        (j-1)^\mu - \alpha j^\mu}$.
  This proves the result.
\end{proof}

%

We have also the following whenever $\bm{c}(t)$ is a solution to
\eqref{eq:BD}:
\begin{lem}
  \label{cor:diffinequalityforG}
  Let $\bm{c}(t)$ be a solution to the Becker-D\"oring equation with
  non-negative, finite density initial datum. Assume \eqref{eq:ai1} and \eqref{eq:bi} to
  hold. Then its associated tail density $\mathcal{G}(t)$ is
  continuously differentiable, and satisfies
  \begin{equation}\label{eq:derivativeofpartial}
    \frac{\d}{\d t}G_j(t)
    = a_{j-1}c_1(t)\pa{G_{j-1}(t)-G_j(t)}+b_j
    \pa{G_{j+1}(t)-G_{j}(t)},
    \qquad j\geq 1.
  \end{equation}
  In particular, if there exist $t_{0} >0$ and $\omega > 0$ such that
  $$c_1(t) \leq \omega \qquad \forall t \geq t_{0},$$ then 
  \begin{equation}
    \label{eq:diffinequalityforG}
    \frac{\d}{\d t}G_j(t) \leq a_{j-1}\omega\pa{G_{j-1}(t)-G_j(t)}+b_j \pa{G_{j+1}(t)-G_{j}(t)},\qquad \forall t\geq t_{0}.
  \end{equation}
\end{lem}

\begin{proof}

  We notice that for any $k,N\in\N$ with $1 <k \leq N$
  $$\sum_{i=k}^N\abs{\frac{\d}{\d t}c_i(t)}= \sum_{i=k}^N\big|a_i c_1(t)c_i(t) - b_{i+1}c_{i+1}(t)-a_{i-1}c_{1}(t)c_{i-1}(t)+b_{i}c_{i}(t)\big| $$
  $$\leq 2\overline{a}\varrho \sum_{i=k-1}^N i^ \gamma c_i(t) + 2\overline{b}\overline{a} \sum_{i=k}^{N+1} i^{\gamma}c_i(t) \leq C\sum_{i=k-1}^{N+1} i^\gamma c_i(t)$$
where we used \eqref{eq:ai1} and \eqref{eq:bi}.   
 Recalling now that  $\sum_{i=1}^\infty ic_i(t)$ converges uniformly on any interval thanks to Proposition 3.1 in \cite{BCP86}, we deduce that $\sum_{i=1}^\infty \frac{\d}{\d t}c_i(t)$ converges
  uniformly on any interval for any $\gamma \in [0,1]$. Thus, since $\bm{c}(t)$ is continuously differentiable, so
  it $\mathcal{G}(t)$, and
  $$\frac{\d}{\d t}G_j(t)
  = \sum_{i=j}^\infty \frac{\d}{\d t}c_i(t)
  = \sum_{i=j}^\infty \pa{W_{i-1}(t)-W_{i}(t)} = W_{j-1}(t)$$
  completing the proof.   The second assertion follows immediately from
  \eqref{eq:derivativeofpartial} and the non-negativity of all the
  elements involved.
\end{proof}

Looking at inequality \eqref{eq:diffinequalityforG} we notice that the
infinite system of differential inequalities for tail densities can be
represented by an infinite constant matrix with entries only in the
diagonal, and above and below it. Moreover, the off--diagonal entries
are non-negative. Unsurprisingly, this will entail a \emph{maximum
  principle} to the system.

For any given vectors $\bm{u},\bm{v}\in \R^n$, we denote by
$\bm{u} \leq \bm{v}$ the case where $u_i \leq v_i$ for all
$i=1,\dots,n$. Given $z \in \R$, we also denote $z_{+}=\max(z,0)$ and,
if $\bm{u}=(u^{1},\ldots,u^{n}) \in \R^{n}$, we set
$\bm{u}_{+}=(u^{1}_{+},\ldots,u^{n}_{+}).$

\begin{lem}[\textbf{\textit{Maximum principle for linear ODE systems}}]
  \label{lem:max-ode}
  Let $T > 0$ and consider the vector of continuously differentiable
  functions $\bm{u}=(u_1, \dots, u_n) \: [0,T) \to \R^n$. Assume
  that \begin{equation}
    \label{eq:vector-dif-ineq}
    \ddt \bm{u}(t) \leq A \bm{u}(t) \qquad \text{ for all }  t \in [0,T),
  \end{equation}
  where $A$ is a constant $n \times n$ matrix whose off-diagonal
  entries are non-negative. Then, if $\bm{u}(0) \leq 0$ we have that
  $\bm{u}(t) \leq 0$ for all $t \in [0,T)$.
\end{lem}

\begin{proof}
  Let $0 \leq t <T$ be given. Since $\bm{u}$ is differentiable at $t$
  we have that for $0<s<T-t $
  \begin{equation*}
    \bm{u}(t+s) \leq \bm{u}(t) + s A \bm{u}(t) + o(s)
    = (I + s A) \bm{u}(t) + o(s).
  \end{equation*}
 Set $A=(a_{i,j})_{i,j=1,\ldots,n}$ and call $s_0 := \inf_{i=1,\dots,n} |a_{i,i}|^{-1} \in (0,+\infty]$. As
  the off-diagonal entries of $A$ are
  non-negative we find that for $0 < s < s_* := \min\{ s_0, T-t \}$,
  all the entries of the matrix $I+sA$ are
  non-negative. Thus, \begin{equation*} [\bm{u}(t+s)]_+ \leq [(I + s
    A) \bm{u}(t)]_+ + o(s) \leq (I + s A) [\bm{u}(t)]_+ + o(s)
  \end{equation*}
  for all $0 < s < s_*$. Denoting by $y(t)$ the $\ell_{1}$-norm of
  $[\bm{u}(t)]_{+}$, i.e. $y(t) := \sum_{j=1}^n u^j_{+}(t)$, we see
  that
  \begin{equation*} 
    y(t+s) \leq
    y(t) + s C y(t) + o(s),
  \end{equation*}
  where 
$$C=\max_{i=1,\dots,n} \sum_{j=1}^n \abs{a_{i,j}}.$$  
Dividing by $s$ and
  taking the limit as $s \to 0$ we see that
  \begin{equation*}
   \liminf_{s\rightarrow 0^+}  \frac{y(t+s)-y(t)}{s}  \leq
    C y(t).
  \end{equation*}
  A generalised version of Gronwall's lemma, following from a generalised comparison theorem that can be found in Lemma 16.4, p. 215 of \cite{Amann90}, implies that
$$y(t) \leq y(0) \re\pa{Ct}.$$
Since $y(0)=0$ we conclude the proof.  
\end{proof}
\begin{rem}
  The above proof is a simple version of the invariance of the cone of
  points with non-positive coordinates using the so-called
  \textit{sub-tangent condition} (as given for example in Theorem 16.5, p. 215 of 
  \cite{Amann90}). A matrix with non-negative
  off--diagonal entries is known as a \emph{Metzler matrix}, and its
  sign-preserving properties are well-known. We have given a full
  proof for the sake of completeness, and since we need the result
  when we deal with an inequality (and not an equality).
\end{rem}

In order to use this maximum principle we define the notion of
\emph{supersolution} for the Becker-Döring equations:

\begin{dfn}[\textit{\textbf{Supersolution}}]\label{def:super}
  Let $0 < \varrho$ and $0 <\omega $ be given. We
  say that a non-negative sequence $(r_j)_{j \geq 1}$ is a
  \emph{$(\omega,\varrho)$-supersolution} to Becker-D\"oring equations
  if
  \begin{enumerate}
  \item $r_{1}\geq \varrho$ 
    \item For all $j \geq 2$ it holds that
    \begin{equation}
      \label{eq:ri-condition}
        a_{j-1} \omega (r_{j-1} - r_j) +b_j (r_{j+1} - r_{j})\leq 0.
    \end{equation}
  \end{enumerate}
\end{dfn}

\begin{rem}
  Notice that, strictly speaking, a sequence
  $(r_{j})_{j\geq 1}$ with the above properties is not a supersolution
  to \eqref{eq:BD} (in the classical ODEs sense) but rather a
  supersolution of the system:
  \begin{equation}\label{eq:ddtx}
    \ddt x_j(t) =  a_{j-1}\omega\pa{x_{j-1}(t)-x_j(t)}+b_j
    \pa{x_{j+1}(t)-x_{j}(t)}, \qquad j \geq 1
  \end{equation}
  with $x_i(t)\leq \varrho$ for all $i \geq 1,$ $t \geq0$. Notice also
  that, whenever \eqref{eq:diffinequalityforG} holds true,
  $\mathcal{G}(t)$ is a subsolution of \eqref{eq:ddtx} on
  $[t_{0},\infty)$.
\end{rem}

The values of $\omega$ and $\varrho$ that are helpful to obtain a maximum principle are connected to those of
$c_1(t)$ and the mass of $\bm{c}$ in the following way:

\begin{prp}[\textit{\textbf{Maximum principle}}]
  \label{lem:maximum-principle}
  Let $\bm{c}(t) = (c_i(t))_{i \geq 1}$ be a solution to the
  Becker-D\"oring equations with non-negative initial condition
  $\bm{c}(0)$. Assume that \eqref{eq:ai1} and \eqref{eq:bi} hold and
  that the density of $\bm{c}$ is $0 < \varrho < \varrho_s$. Let
  $\mathcal{G}(t)$ denote the tail density of $\bm{c}(t).$ Take
  $\omega >0$ and $0 \leq t_0 < t_1$, and denote $I := [t_0,
  t_1]$. Assume that
  $$c_1(t) \leq
  \omega \qquad \text{ for all }
  \quad t \in I.$$
  Let $(r_j)_{j \geq 1}$ be a $(\omega,\varrho)$-supersolution to the
  associated Becker-D\"oring equations. Then if
  \begin{equation*}
    G_j(t_0) \leq r_j \qquad \text{for all $j \geq 1$,}
  \end{equation*}
  we find that
  \begin{equation*}
    G_j(t) \leq r_j
    \qquad \text{for all $t \in [t_0, t_1]$ and all $j \geq 1$.}
  \end{equation*}
\end{prp}

\begin{proof}
  Since $[t_{0},t_{1}]$ is compact and the sequence $\mathcal{G}(t)=(G_{j}(t))_{j \geq 1}$ is a non-increasing sequence of continuous functions that converge pointwise to zero, we conclude from Dini's Theorem that
  $$\lim_{j\to \infty}\sup_{t \in [t_{0},t_{1}]}G_{j}(t)=0.$$
 Given $\varepsilon >0$, set 
 $$H_{j}(t)=G_{j}(t)-r_{j}-\varepsilon, \qquad \forall t \in [t_{0},t_{1}], \quad j \geq 2.$$
 There exists an $M \geq 1$, independent in $t$, such that
 \begin{equation}
   \label{eq:maximum-large-j}
   H_{j+1}(t) \leq 0
   \qquad \forall t \in [t_{0},t_{1}], \quad j \geq M.
 \end{equation}
 In addition, as $H_1(t) = \varrho - r_1 -\varepsilon$, the condition $\varrho \leq r_1$ implies that $H_1(t) <0$.\\
 Lemma \ref{cor:diffinequalityforG} and condition \eqref{eq:ri-condition} for the supersolution sequence imply that
 \begin{equation}\label{eq:Hj}
    \ddt H_j(t)
    \leq b_j (H_{j+1}(t) - H_{j}(t))
    + a_{j-1} \omega (H_{j-1}(t) - H_j(t))
    \qquad \forall j \geq 2
  \end{equation}
  on $I$. Due to \eqref{eq:maximum-large-j} and the fact that $H_1(t)<0$ we can consider the system \eqref{eq:Hj} for $j= 2,\ldots, M$ only. This system can be rewritten as
  \begin{equation*}
    \ddt \left(
      \begin{array}{c}
        H_2(t)\\H_{3}(t) \\\vdots \\ \vdots \\H_{M-1}(t)\\H_{M}(t)
      \end{array}
    \right)
    \leq
    \left(
      \begin{array}{ccccc}
        -\alpha_{2} & b_{2} & 0 & \cdots & 0 \\
        a_{2} \omega  & -\alpha_{3} & b_{3} & \cdots & 0  \\
        \cdots & \cdots & \cdots & \cdots & \cdots\\ 
        \cdots & \cdots & \cdots & \cdots & \cdots\\
        0 & \cdots & a_{M-2}\omega & -\alpha_{M-1} & b_{M-1} \\
        0 & \cdots & 0 & a_{M-1} w & -\alpha_{M}
      \end{array}
    \right)
    \left(
      \begin{array}{c}
        H_2(t) \\H_{3}(t)
        \\\vdots \\ \vdots
        \\ H_{M-1}(t) \\H_{M}(t)
      \end{array}
    \right),
  \end{equation*}
  where $\alpha_{j} = a_{j-1} \omega + b_{j}$. As all the
  off--diagonal entries of the above matrix are non-negative and since
  our initial conditions imply
  $$H_j(t_0)=G_j(t_0)-r_j - \varepsilon <0,$$   we find that due to our maximum principle (Lemma \ref{lem:max-ode})
  $$H_j(t) \leq 0 \qquad  \qquad \forall \; 2\leq j \leq M, t\in I.$$
  Together with the bounds on $H_1$ and $(H_{j})_{j\geq M+1}$ we conclude that on $I$
  $$G_{j}(t) \leq r_j + \varepsilon.$$
  As $\varepsilon$ was arbitrary, we get our desired result.
\end{proof}

As we can see, if $\varrho$ is the density of $\bm{c}(t)$, the
associated $(\omega,\varrho)$-supersolution will control
$\mathcal{G}(t)$, with appropriate initial conditions. The question
remains as to which $\omega$ one may choose. This choice will be
crucial to \emph{the existence of a supersolution} that bounds
$\mathcal{G}(t)$ at a suitable time. Since in the subcritical case
$c_1(t)$ converges to $\overline{z} < z_s$, it seems natural to choose
$\omega$ close to, but larger than, $\overline{z}$. This is indeed the
required ingredient to construct a supersolution. The following lemma,
which not only gives us the existence of a supersolution but also
gives us moment connections between it and $\mathcal{G}$, is
reminiscent to Lemma 3.4 in \cite{Canizo2005Asymptotic}.

\begin{lem}
  \label{lem:finding-supersolutions}
  Assume that conditions \eqref{eq:ai1} to \eqref{eq:Qi-limit} hold.  Let
  $0 < \varrho < \varrho_s$ and $0 < \omega < z_s$ be given. Consider
  a non-negative, non-increasing sequence $({g}_j)_{j \geq 1}$ that
  tends to $0$ as $j$ goes to infinity and such that
  $g_1 \leq \varrho$. Then, there exists a
  $(\omega, \varrho)$-supersolution $(r_j)_{j \geq 1}$ to the
  associated Becker-D\"oring equations which tends to $0$ as $j$ goes
  to infinity, and satisfies
  \begin{equation*}
    g_j \leq r_j \qquad \forall j \geq 1.
  \end{equation*}
  Moreover, $(r_j)_{j \geq 1}$ can be chosen so that for any $1 \leq \delta < z_s/\omega$ and any positive, eventually non-decreasing sequence $(\phi_j)_{j \geq 1}$ satisfying
  \begin{equation}
    \label{eq:phi-decay}
    \limsup_{j \to +\infty} \frac{\phi_{j}}{\phi_{j-1}}
    \leq
    \delta,
  \end{equation} 
  we have that  
  \begin{equation}
    \label{eq:supersol-moment-bound}
    \sum_{j=1}^\infty \phi_j r_j
    \leq
    C \left( 1 + \sum_{j=1}^\infty \phi_j g_j \right)
  \end{equation}
  where $C > 0$ is a fixed constant that depends only on $(\phi_j)_{j\geq 1}$, $\varrho$, $\omega$, $\delta$ and the coefficients $(a_i)_{i \geq 1}$,
  $(b_i)_{i \geq 1}$.
\end{lem}

\begin{proof}
  According to \eqref{eq:Qi} and \eqref{eq:Qi-limit}, one notices that
  $\lim_{j\to\infty}b_j / a_{j-1}= z_s$, we can find
  $1 < \lambda \in (\delta,\frac{z_{s}}{\omega})$ and $N \geq 1$ such
  that
  \begin{equation*}
    b_{j} \geq \lambda \, \omega \,a_{j-1}
    \qquad \forall j \geq N.
  \end{equation*}
  For $j \geq N$, we set
 \begin{equation*}\begin{cases}
    h_j &:= g_j - g_{j+1} \geq 0
    \\
    s_{N} &:= \frac{\varrho}{\lambda \omega}+ h_{N},
    \quad
    s_{j+1} := \max\left\{
      \dfrac{s_j}{\lambda}, h_{j+1}
    \right\}\end{cases}\end{equation*}
  and define
  \begin{equation*}
    {r}_j := \sum_{\ell = j}^\infty s_\ell.
  \end{equation*}
We will now show that this sequence is well defined and is bounded. Indeed, using the fact that
$$0 \leq h_j \leq g_j \leq g_1 \leq \varrho \qquad  \forall j \geq 1,$$
and the fact that $\frac{\varrho}{\lambda \omega} \leq s_N \leq \varrho\pa{1+\frac{1}{\lambda \omega}}$ and $1 \leq\delta <\lambda$, we can use a simple induction to show that
  \begin{equation*}
    \label{eq:sj-range}
    0 < s_j \leq \varrho\pa{1+\frac{1}{\lambda \omega}}
    \qquad \forall j\geq N.
  \end{equation*}
  Moreover, as $s_{j+1} \leq \frac{s_{j}}{\lambda}+h_{j+1}$ for all
  $j \geq N$, we have that for any $p \geq N$
  \begin{equation*}\label{eq:s_{N}}\begin{split}
    \sum_{j=N}^{p+1}s_{j}&=s_{N}+ \sum_{j=N}^{p}s_{j+1}
    \leq \frac{\varrho}{\lambda \omega}+h_{N}+\sum_{j=N}^{p}h_{j+1}
    + \frac{1}{\lambda}\sum_{j=N}^{p}s_{j}
    \\
    &\phantom{+++}
    \leq \frac{\varrho}{\lambda \omega}+g_{N+1}-g_{p+2} +\frac{1}{\lambda}\sum_{j=N}^{p+1}s_{j}.
  \end{split}\end{equation*}
Due to the non-negativity of $g_j$ and the fact that $1<\lambda$ we
conclude that
$$\sum_{j=N}^{p+1} s_{j} \leq \frac{g_{N+1}+\frac{\varrho}{\lambda \omega}}{\left(1-\frac{1}{\lambda}\right)}.$$ 
As $p$ is arbitrary, this shows that the sum converges and thus that
$(r_j)_{j\geq N}$ is well defined with $\lim_{j\to\infty}r_{j}=0.$
Moreover, using that $g_{N+1} \leq \varrho$ we see from the previous
inequality that
$$r_j \leq \frac{\varrho(\lambda\omega+1)}{\omega(\lambda-1)} \qquad j\geq N.$$
From its definition, $(r_j)_{j\geq N}$ is clearly non-negative and non-increasing. In addition
$$r_j \geq  \sum_{\ell = j}^\infty h_\ell = g_j,$$
where we used the fact that $(g_j)_{j\geq 1}$ goes to zero as $j$ goes to infinity. Due to the choice of $N$, we have that for all $j\geq N+1$
 $$\frac{r_{j-1}-r_{j}}{r_{j}-r_{j+1}}=\frac{s_{j-1}}{s_{j}}
   \leq \lambda \leq \frac{b_{j}}{\omega a_{j-1}}.$$
   All of the above show that we have managed to construct a
   supersolution to the associated Becker-D\"oring equation from the
   point $j=N+1$. We are left with defining it for $j<N$ and to check the supersolution condition for $j\leq N$. We set for any
   $j<N$
  \begin{equation}
    \label{eq:def-rj-<N}
    r_{j} = \max\{\varrho, r_{N}\}.
  \end{equation}
  Clearly, by its definition
  $$g_j  \leq g_1 \leq \varrho \leq r_j,$$
  which also shows that $ \varrho \leq r_1$. In addition, one checks that
$$a_{j-1}\omega (r_{j-1}-r_j)+b_j(r_{j+1}-r_{j})=\begin{cases}
0 & j<N-1 \\
b_{N-1}\pa{r_N-\max\{\varrho, r_{N}\}} \leq 0 & j=N-1 \\
a_{N-1}\pa{\max\{\varrho, r_{N}\}-r_N} - b_{N}s_N \leq 0 & j=N 
\end{cases}.$$
Indeed, the last inequality is valid since, due to the choice of $N$,
$$a_{N-1}\pa{\max\{\varrho, r_{N}\}-r_N} - b_{N}s_N \leq b_N\pa{\frac{\varrho}{\lambda \omega} - s_N} \leq 0.$$
Thus $(1)$ and $(2)$ from Definition \ref{def:super}, are satisfied up
to $j=N$. Together with our definition for $j\geq N$ we conclude
that $(r_{j})_{j\geq 1}$ is an $(\omega,\varrho)$--supersolution to
the associated Becker-D\"oring equation. Moreover,
$$r_j \leq \varrho\max \{1, \frac{\lambda\omega+1}{\omega(\lambda-1)} \}=\frac{\varrho(\lambda\omega+1)}{\omega(\lambda-1)}
, \qquad \forall j \geq 1.$$
 We turn our attention now to the second part of the proof. Due to the conditions on $(\phi_j)_{j\geq 1}$ we can find $\delta < \delta_* < \lambda$ and  $M\geq N \geq 1$ such that for all $j\geq M$
$$\phi_{j-1} \leq \phi_{j}$$
 $$\frac{\phi_j-\phi_{j-1}}{\phi_j} \leq 1-\frac{1}{\delta_*}.$$
 Consider the sum $\sum_{j=1}^\infty r_j \phi_j$. Using again that
 $s_{j+1} \leq \frac{s_j}{\lambda}+h_{j+1}$ for any $j\geq M \geq N$,
 we have that
 $$s_j \leq r_j = \sum_{\ell=j}^\infty s_\ell \leq s_j +
 \frac{1}{\lambda}\sum_{\ell=j}^\infty s_\ell + \sum_{\ell=j}^\infty
 h_{\ell+1} \leq s_j+\frac{r_j}{\lambda} + g_{j+1},
 \qquad j \geq M.$$
Thus
$$s_j \leq r_j \leq \frac{\lambda (s_j+g_{j+1})}{\lambda-1}.$$
From the above we can estimate that
\begin{multline*}\sum_{j=M}^\infty \phi_j r_j \leq \frac{\lambda}{\lambda-1}\pa{\sum_{j=M}^\infty \phi_{j} s_j+\sum_{j=M+1}^\infty \phi_{j-1} g_j}\\ \leq \frac{\lambda}{\lambda-1}\pa{\sum_{j=M}^\infty \phi_{j} \pa{r_j-r_{j+1}}+\sum_{j=M+1}^\infty \phi_{j} g_j }\\
= \frac{\lambda}{\lambda-1}\pa{\sum_{j=M}^\infty r_j \pa{\phi_{j}-\phi_{j-1}} +\phi_{M-1}r_{M} +\sum_{j=M+1}^\infty \phi_{j} g_j }\\
\leq \pa{1+\frac{1}{\lambda-1}}\pa{1-\frac{1}{\delta_*}}\sum_{j=M}^\infty r_j\phi_j  +\frac{\lambda}{\lambda-1}\pa{\phi_{M-1}r_{M}+\sum_{j=M+1}^\infty \phi_{j} g_j},
\end{multline*}
 which implies that
\begin{equation*}\label{eq:momentsrj}
\sum_{j=M}^\infty \phi_j r_j \leq  \frac{\lambda \delta_*}{\lambda-\delta_*}\pa{\phi_{M-1} r_{M}+\sum_{j=M+1}^\infty \phi_{j} g_j }.
\end{equation*}
$$\leq \frac{\lambda \delta_*}{\lambda-\delta_*}\pa{\frac{\varrho(\lambda\omega+1)}{\omega(\lambda-1)}\phi_{M-1}+\sum_{j=M+1}^\infty \phi_{j} g_j}$$
 Thus, as
$$\sum_{j=1}^{M-1} \phi_j r_j \leq\frac{\varrho(\lambda\omega+1)}{\omega(\lambda-1)}\sum_{j=1}^{M-1} \phi_j$$
 we conclude that
 $$\sum_{j=1}^\infty \phi_j r_j  \leq C\pa{1+\sum_{j=1}^\infty \phi_{j} g_j}$$
 where 
$$C=2\max\pa{\frac{\varrho(\lambda\omega+1)}{\omega(\lambda-1)}\sum_{j=1}^{M-1} \phi_j,\frac{\lambda \delta_*}{\lambda-\delta_*}\max\pa{1,\frac{\varrho(\lambda\omega+1)}{\omega(\lambda-1)}\phi_{M-1}}}.$$
 This completes the proof.
\end{proof}

\begin{rem}\label{rem:choiceofphi}
It is important to note that the sequences $(\phi_{j})_{j\geq 1}$ given by
$$\phi_j=j^k \quad (k\geq 0), \qquad \text{ or } 	\qquad \phi_j = \re\pa{j^\mu} \quad (0<\mu<1), \quad j \geq 1$$
both satisfy condition \eqref{eq:phi-decay} with $\delta=1$. This
means that we can build a supersolution with comparable moments, and
stretched exponential moments, to those of $\mathcal{G}(t)$. This will
be a crucial element in the proof of our main theorems. Note that
$\phi_{j}=\re\pa{\eta\,j}$ $(j \geq 1)$ is also allowed, as long as
$\eta < \log\pa{\frac{z_s}{\omega}}$.
\end{rem}

With these tools at hand, we have the main ingredient to prove our main theorem.

\section{On the propagation of moments}
\label{sec:moments}

From the previous section we know that as long as $c_1(t)<z_s$ in a
certain time interval, we are able to construct a supersolution to the
associated Becker-D\"oring equations whose moments are strongly
related to the moments of $\bm{c}(t)$, thus allowing us to take
advantage of the maximum principle in Lemma
\ref{lem:maximum-principle} to obtain a uniform bound. However, the
condition on $c_1(t)$ is not necessarily valid at all times. Before we
prove our main theorems, we show that one can find an explicit time,
$T_0\geq 0$, such that for all $T>T_0$, $c_1(t) <z_s$. Before that
time the moments, and stretched exponential moments, grow at most
exponentially in time (which was already noted in previous works).

The fact that $c_1(t) < z_s$ after a certain time $T_0$ is a
consequence of a stronger statement about the convergence to
equilibrium of the solution to the Becker-D\"oring equations. The
quantitative version we state here uses the relative free energy
mentioned in the introduction and can be easily deduced from results
in \cite{Canizo2007Convergence} and \cite{Canizo2015Trend}:

\begin{thm}
  \label{thm:quantitative-convergence}
  Consider the Becker-D\"oring equations with coagulation and
  fragmentation coefficients $(a_i)_{i\geq1},(b_i)_{i\geq1}$ such that
conditions \eqref{eq:ai1} to \eqref{eq:critical-eq} hold. Assume that
  $\bm{c}(t) = (c_i(t))_{i \geq 1}$ is a solution to the
  Becker-D\"oring equations with non-negative, subcritical initial
  datum $\bm{c}(0)$ satisfying \eqref{eq:c0-moments}. Then, there
  exists a constant $C>0$, depending only on the coagulation and
  fragmentation coefficients, the density $\varrho$, the initial
  moment of $\bm{c}$ of order $\max\{2-\gamma, 1+\gamma\}$ and the
  initial relative free energy, such that
  \begin{equation}
    \label{eq:quantitative-convergence}
    \sum_{i=1}^\infty i\abs{c_i(t) - \Q_i} \leq \frac{C}{\sqrt{1+\abs{\log t} }}, \qquad \forall t > 0.
  \end{equation}
\end{thm}

\begin{proof}
  The result in \cite{Canizo2007Convergence} is valid under
  \eqref{eq:ai1}--\eqref{eq:c0-moments} for $\gamma \in [0,1)$ (and
  holds true for more general discrete coagulation models). For
  $\gamma=1$ (i.e., the second option in \eqref{eq:ai1}), the rate is
  actually exponential thanks to a recent result by the
  authors: see Theorem 1.3 in \cite{Canizo2015Trend}. Notice that for
  the case $\gamma=1$ no assumptions on the propagation of moments are
  needed in \cite{Canizo2015Trend}.
\end{proof}

As a consequence we have:

\begin{cor}\label{cor:timeforc1}
  Assume the conditions of Theorem
  \ref{thm:quantitative-convergence}. Then for any $\delta>0$ one has
  $$c_1(t)<\overline{z}+\delta \qquad \text{ for all } t>T_\delta,$$ 
  where $T_\delta=\max\pa{1,\re\pa{\frac{C^2}{\delta^2}-1}}$, with
  $C>0$ is the explicit constant from Theorem
  \ref{thm:quantitative-convergence}.
\end{cor}

\begin{proof}
  Since
  $$\abs{c_1(t)-\overline{z}} \leq  \sum_{i=1}^\infty i\abs{c_i(t) - \Q_i} \leq \frac{C}{\sqrt{1+\abs{\log t} }}$$
the result follows immediately. 
\end{proof}
The last issue that we need to deal with before being able to tackle
our main theorem is the issue of the possible growth of our moments,
and stretched exponential moments, in the time until $c_1(t)$ is in
the right range to use our machinery from the Section \ref{sec:maximum}. This
has actually been shown in \cite{BCP86}:

\begin{lem}
  \label{lem:short-time}
  Consider the Becker-D\"oring equations with coagulation and
  fragmentation coefficients $(a_i)_{i\geq1},(b_i)_{i\geq1}$ such that
 \eqref{eq:ai1} and \eqref{eq:bi} hold true. Let $(\phi_{i})_{i\geq 1}$ be a
  non-negative sequence such that
  \begin{equation}
    \label{eq:ph}
    \begin{cases}
      \phi_{i+1}-\phi_{i} \geq \varepsilon \phi_{1} \qquad \forall i \geq 1 \\
      \sup_{i \geq 1}\dfrac{a_{i}\left(\phi_{i+1}-\phi_{i}\right)}{\phi_{i}}=\A_{\phi} < \infty,
    \end{cases}
  \end{equation}
  for some $\varepsilon > 0$. Then, if $\bm{c}(t)$ is the solution to
  the Becker-D\"oring equations with non-negative initial datum
  $\bm{c}(0)$ and density $\varrho$ such that
  $$M_{\phi}(\bm{c}(0)) := \sum_{i=1}^\infty \phi_i c_{i}(0) < \infty,$$
  there exists a positive constant $C_{\phi}$ depending on
  $\A_{\phi},\varepsilon$ and $\varrho$ such that
\begin{equation*}
  M_{\phi}\pa{\bm{c}(t)}:=\sum_{i=1}^{\infty}\phi_{i}c_{i}(t) \leq  \re\pa{C_\phi t}M_{\phi}\pa{\bm{c}(0)}, \qquad \forall t\geq 0.
  \end{equation*}
\end{lem}
\begin{proof} 
  A detailed proof of the result is given in \cite{BCP86}. For the
  sake of completeness we provide a formal proof here, from which a fully
  rigorous one can be obtained by standard approximation
  arguments. Using \eqref{eq:BD-weak} with the sequence
  $(\phi_{i})_{i \geq 1}$ we get
  \begin{equation*}
    \begin{split}
      \ddt \sum_{i=1}^\infty \phi_{i} c_{i}(t)
      &= \sum_{i=1}^\infty  \left(a_i c_1(t) c_i(t)-b_{i+1}c_{i+1}(t)\right)
      (\phi_{i+1}-\phi_{i}-\phi_{1})\\
      &\leq
      \sum_{i=1}^\infty a_i c_1(t) c_i(t)
      (\phi_{i+1}-\phi_{i})
      + \sum_{i=1}^{\infty}b_{i+1}c_{i+1}(t) \phi_{1}
    \end{split}
  \end{equation*}
  where we used that $(\phi_{i})_{i\geq 1}$ is non-negative. The first
  sum is estimated with \eqref{eq:ph}:
  $$\sum_{i=1}^\infty a_i c_1(t) c_i(t) (\phi_{i+1}-\phi_{i}) \leq \A_{\phi}c_{1}(t) \sum_{i=1}^{\infty}\phi_{i}c_{i}(t) \leq \varrho \A_{\phi} \sum_{i=1}^{\infty}\phi_{i}c_{i}(t),$$
  while using \eqref{eq:ph} and \eqref{eq:bi} we find that
  $$\sum_{i=1}^{\infty} b_{i+1}c_{i+1}(t) \phi_{1}
  \leq
  \phi_{1}\overline{b}
  \sum_{i=2}^\infty a_{i}c_{i}(t)
  \leq \varepsilon^{-1}\overline{b}
  \sum_{i=2}^{\infty}a_{i}c_{i}(t)(\phi_{i+1}-\phi_{i})
  \leq  \varepsilon^{-1}\overline{b}\A_{\phi}
  \sum_{i=2}^{\infty}\phi_{i}c_{i}(t).$$
 The result then follows with
 $C_{\phi}=\left(\varrho+\varepsilon^{-1}\overline{b}\right)\A_{\phi}.$
\end{proof}

\begin{rem}\label{rem:momentcontrol}
  The two main types of moments we consider, namely
  \begin{equation*}
    (\phi_j)_{j\geq1}=(j^k)_{j\geq 1}
    \quad (k\geq 1), \quad \text{ and }
    \quad (\phi_j)_{j\geq 1}=\pa{\re\pa{j^\mu}}_{j\geq 1}
    \quad (0\leq  \mu\leq 1-\gamma),
  \end{equation*}
  satisfy the assumptions of Lemma \ref{lem:short-time}. On the
  contrary, the previous lemma cannot be applied to exponential
  moments (with weight $e^{\mu i}$) if $a_i$ diverges to $+\infty$
  with $i$.
\end{rem}

We are now ready to prove our main theorems.
\begin{proof}[Proof of Theorem \ref{thm:moment-propagation}]
  Since we are dealing with a subcritical solution, using Corollary
  \ref{cor:timeforc1} we can find an explicit time $T_0>0$ such that  
  $$c_1(t) < \omega <z_s
  \qquad \text{for any $t>T_0$.}$$
  Due to Lemma \ref{lem:short-time} we can find an explicit constant
  $C_k>0$ such that for all $t\leq T_0$
$$M_{k}(T_{0}) \leq \re\pa{C_k t} M_{k}(0).$$
Considering the tail density sequence $\mathcal{G}(t)$, we use Lemma \ref{lem:finding-supersolutions} to find a supersolution to the associated Becker-D\"oring equation for $t \geq T_0$ such that
$$G_j(T_0) \leq r_j \qquad \forall j\geq 1$$
and
  \begin{equation}\label{eq:momentmainI}
  \begin{split}
   \sum_{j=1}^\infty j^{k-1} r_j &\leq C\pa{1+\sum_{j=1}^\infty j^{k-1} G_j(T_0)} \\
   &\leq C\pa{1+M_{k}\pa{\bm{c}(T_0)}} \leq  C\pa{1+ M_k(0) \re\pa{C T_{0}}},
    \end{split}
  \end{equation}
  where we have used Lemma \ref{lem:propertiesofpartial} and \ref{lem:short-time}. According to the maximum principle, Proposition \ref{lem:maximum-principle}, and the fact that $c_1(t) < \omega$ for $t>T_0$ we find that 
  \begin{equation*}
    G_j(t) \leq r_j
    \qquad \text{for all $t \geq T_0$, for all $j \geq 1$},
  \end{equation*}
and thus, using Lemma \ref{lem:propertiesofpartial} again, and \eqref{eq:momentmainI}, we have that for all $t\geq T_0$
  \begin{equation*}
    M_k(t)
    \leq (k+1) \sum_{j=1}^\infty j^{k-1} G_j(t)
    \leq (k+1) \sum_{j=1}^\infty j^{k-1} r_j
    \leq C\pa{1+ M_k(0) \re\pa{C T_{0}}}.
  \end{equation*}
  This concludes he proof.
\end{proof}

\medskip

\begin{proof}[Proof of Theorem \ref{thm:moment-propagation2}]
  We set
  $$\E_{\mu}(t)=\sum_{i=1}^{\infty}
  \re\pa{\alpha\,i^{\mu}}c_{i}(t), \qquad t \geq 0.$$
  The link between $\E_{\mu}(t)$ and the tail density $\mathcal{G}(t)$
  is given by \eqref{eq:doublebound}, namely
  \begin{equation}\label{eq:doublebound1}  
    \eta_{1}\sum_{j=1}^{\infty} \psi_j G_{j}(t)
  \leq
  \E_{\mu}(t)
  \leq
  \eta_{2}\sum_{j=1}^{\infty} \psi_j G_{j}(t),
  \qquad \forall t \geq 0.
  \end{equation}
  for some positive constants $\eta_{1},\eta_{2} >0$ depending only on $\alpha,\mu$ and $(\psi_{j})_{j\geq 1}:=\pa{\alpha \mu j^{\mu-1} \re\pa{\alpha j^{\mu}}}_{j\geq 1}$. At this point we just mimic the proof of Theorem \ref{thm:moment-propagation}: We find an explicit time  $T_0>0$ such that for any $t>T_0$ we have that
$c_1(t) < \omega <z_s.$ Then, using Lemma \ref{lem:finding-supersolutions}, and noting that our $\psi_j$ satisfies its conditions, we find a supersolution to the associated Becker-D\"oring equation, $(r_j)_{j\geq 1}$ such that $G_{j}(T_{0}) \leq r_{j}$ $(j \geq 1)$ and
$$\sum_{j=1}^{\infty}\psi_{j}r_{j} \leq C(1+\sum_{j=1}^{\infty}\psi_{j}G_{j}(T_{0})) \leq C(1+\eta^{-1}_{1}\E_{\mu}(T_{0})).$$
Invoking Proposition \ref{lem:maximum-principle}, we get  $G_j(t) \leq r_j$ for all $t \geq T_{0}$ and all $j \geq 1.$ Using again \eqref{eq:doublebound1}, we have then
$$\E_{\mu}(t) \leq \eta_{2}\sum_{j=1}^{\infty}\psi_{j}G_{j}(t) \leq \eta_{2}\sum_{j=1}^{\infty}\psi_{j}r_{j} \leq C\,\eta_{2}(1+\eta^{-1}_{1}\E_{\mu}(T_{0})), \quad \forall t \geq T_{0}.$$
completing the proof by using Lemma \ref{lem:short-time} with
$\phi_i=\re\pa{\alpha i^\mu}$ $(i \geq1)$.
\end{proof}

\section*{Acknowledgement}
JAC was supported by project MTM2014-52056-P, funded by the Spanish
government and the European Regional Development Fund. AE was
supported by the Austrian Science Fund (FWF) grant M 2104-N32. We
would like to thank André Schlichting for useful comments on our
previous work, which led to us realising that no uniform bounds for
moments were available.


\bibliographystyle{plain}

   \end{document}